\documentclass[11pt]{article}
\usepackage[utf8]{inputenc}

\usepackage{hyperref}
\usepackage{amsmath,amsthm,amssymb,caption,subcaption,nicefrac}
\usepackage{cleveref}
\usepackage{fullpage}
\usepackage{bm}
\usepackage{algorithm}
\usepackage[noend]{algorithmic}
\usepackage[table]{xcolor}
\usepackage{multirow}

\newtheorem{theorem}{Theorem}

\newtheorem{lemma}[theorem]{Lemma}
\newtheorem{corollary}[theorem]{Corollary}
\newtheorem{fact}[theorem]{Fact}
\newtheorem{definition}[theorem]{Definition}
\newtheorem{observation}[theorem]{Observation}

\newcommand{\cM}{\mathcal{M}}
\newcommand{\cA}{\mathcal{A}}
\newcommand{\cI}{\mathcal{I}}

\newcommand{\ind}{\mathbf{1}}
\newcommand{\eps}{\epsilon}
\newcommand{\N}{\mathbb{N}}
\newcommand{\Z}{\mathbb{Z}}
\newcommand{\R}{\mathbb{R}}

\newcommand{\E}{\mathbb{E}}
\newcommand{\cD}{\mathcal{D}}
\newcommand{\cJ}{\mathcal{J}}

\newcommand{\bv}{\mathbf{v}}
\newcommand{\bu}{\mathbf{u}}
\newcommand{\by}{\mathbf{y}}
\newcommand{\bz}{\mathbf{z}}
\newcommand{\bp}{\mathbf{p}}
\newcommand{\bq}{\mathbf{q}}
\newcommand{\bhp}{\hat{\bp}}

\renewcommand{\phi}{\varphi}
\newcommand{\bphi}{\bm{\phi}}

\newcommand{\hn}{\hat{n}}
\DeclareMathOperator{\Geo}{Geo}
\DeclareMathOperator{\argmin}{argmin}

\allowdisplaybreaks

\title{Tight Bounds for Differentially Private Anonymized Histograms}
\author{Pasin Manurangsi\thanks{Google Research. Email address: \texttt{\small pasin@google.com}.}
}
\date{}

\begin{document}

\maketitle


\begin{abstract}
In this note, we consider the problem of differentially privately (DP) computing an \emph{anonymized histogram}, which is defined as the multiset of counts of the input dataset (without bucket labels). In the low-privacy regime $\eps \geq 1$, we give an $\eps$-DP algorithm with an expected $\ell_1$-error bound of $O(\sqrt{n} / e^\eps)$. In the high-privacy regime $\eps < 1$, we give an $\Omega(\sqrt{n \log(1/\eps) / \eps})$ lower bound on the expected $\ell_1$ error. In both cases, our bounds asymptotically match the previously known lower/upper bounds due to~\cite{Suresh19}.
\end{abstract}

\section{Introduction}

In the \emph{histogram} problem, each user $i$’s input $x_i$ is one of the buckets from a universe $[B] := \{1, \dots, B\}$. The goal is to compute the count of every bucket $b$, i.e. $n_b := |\{i \mid x_i = b\}|$. We would like to estimate the counts in such a way that user’s privacy is respected under the notation of \emph{differential privacy (DP)}~\cite{dwork2006calibrating,dwork2006our}:

\begin{definition}[Differential Privacy~\cite{dwork2006calibrating}]
For $\eps, \delta > 0$, a randomized algorithm $\cA$ is $(\eps, \delta)$-DP if, for any set $S$ of outputs and any pair of neighboring datasets $\cD, \cD'$, we have $\Pr[\cA(\cD) \in S] \leq e^{\eps} \cdot \Pr[\cA(\cD') \in S] + \delta$. When $\delta = 0$, we write $\eps$-DP as a shorthand for $(\eps, 0)$-DP.
\end{definition}

Here $\cD, \cD'$ are said to be neighbors if $\cD'$ results from adding or removing a single user from $\cD$.

Differentially private histogram is one of the most well-studied problems in DP and (asymptot- ically) optimal errors are known across a variety of parameter settings and models of differential privacy (e.g.~\cite{dwork2006calibrating,MishraS06,KorolovaKMN09,HardtT10,HsuKR12,ErlingssonPK14,BassilyS15,BunNS19,BalcerC20,BalcerCJM21,GKM21}).

In many applications, however, we are not interested in the \emph{buckets} themselves. In these scenarios, it suffices to consider the \emph{anonymized histogram} (aka \emph{histogram of histogram}) of the dataset which is defined as $(n^{(1)}, n^{(2)}, \cdots, n^{(B)})$ where $n^{(i)}$ denote the $i$-th largest number among the counts $n_b$’s. For example, if the input dataset is $(x_1,x_2,x_3,x_4,x_5) = (1,1,3,2,3)$, then the histogram is $(n_1,n_2,n_3) = (2,1,2)$ whereas the anonymized histogram is $(n^{(1)},n^{(2)},n^{(3)}) = (2,2,1)$. It is known that anonymized histogram is sufficient to estimate symmetric properties of discrete distributions~\cite{AcharyaDOS17,Suresh19}; it also has applications in releasing password frequency lists~\cite{BlockiDB16} and degree sequences of graphs~\cite{KarwaS12}.

Throughout this note, we focus on the problem of DP anonymized histogram under the $\ell_1$-error, which is defined as $\sum_i \|n^{(i)} - \hn^{(i)}\|$ where $(n^{(i)})_i$ denotes the true anonymized histogram and $(\hn^{(i)})_i$ denotes the output estimate. We note that other measures of error have also been studied for DP anonymized histograms; see e.g.~\cite{HayLMJ09,HayRMS10} which focus on $\ell^2_2$ error. Similar to previous works, we focus on the regime where the universe $B$ is unbounded\footnote{This is only used for lower bounds; when $B$ is small, computing DP histogram directly yields errors of $O(|B|/\eps)$ in the high-privacy regime and $O(|B|/e^\eps)$ in the low-privacy regime which can be smaller than the bounds below.}.

Despite being a natural variant of the histogram problem, the optimal errors achievable by DP algorithms for anonymized histogram are not yet fully understood. Blocki et al.~\cite{BlockiDB16} initiated the study of DP anonymized histogram under the $\ell_1$ error. They observe that the problem is equivalent to that of privatizing \emph{integer partitions} (see \Cref{sec:overview} for more detail). Since there are only $\exp(O(\sqrt{n}))$ distinct integer partitions~\cite{HR18}, Blocki et al. showed that one can achieve an expected error of $O(\sqrt{n}/\eps)$ via the $\eps$-DP exponential mechanism~\cite{McSherryT07}. Generally, the exponential mechanism is not efficient; nonetheless, Blocki et al. provided a variant of the algorithm that is $(\eps, \delta)$-DP and achieves a similar expected error but runs in time $poly(n, \log(1/\delta), 1/\eps)$. Subsequently,
Ald{\`{a}} and Simon~\cite{AldaS18} show that any $\eps$-DP algorithm for $1/8 > \eps > \Omega(\frac{1}{n})$ must incur an error of $\Omega(\sqrt{n/\eps})$. Recently, Suresh~\cite{Suresh19} proved two arguably surprising results. First, he provides an $\eps$-DP algorithm for the “high-privacy” regime $\eps \leq 1$ with error $O(\sqrt{n \log(1/\eps)/\eps})$, which in light of the aforementioned lower bound is tight up to a factor of $O(\sqrt{\log(1/\eps)})$. Furthermore, he devises an algorithm with error $O(\sqrt{n}/e^{c\eps})$ where $c$ is a small constant in the “low-privacy” regime $\eps > 1$; in this setting, he also shows a lower bound of $\Omega(\sqrt{n}/e^\eps)$. His upper and lower bounds in the low-privacy regime match up to a constant factor in front of $\eps$, although one should note that this still leaves a (potentially large) multiplicative gap of $O(e^{(1-c)\eps})$ between the two bounds.

\section{Our Contributions}

Our main contribution is closing these two gaps left from~\cite{Suresh19}. Specifically, in the low-privacy regime $\eps > 1$, we give an algorithm with expected error of $O(\sqrt{n}/e^{\eps})$, while in the high-privacy regime $\eps \leq 1$, we prove a lower bound of $\Omega(\sqrt{n \log(1/\eps)/\eps})$ on the expected error. Both of which are tight up to constant multiplicative factors. The latter also negatively answers an open question of Ald{\`{a}} and Simon~\cite{AldaS18} who asked whether the bound of $\Omega(\sqrt{n/\eps})$ proved in their paper is tight. We also remark that, while previous works focus their efforts on proving lower bounds for $\eps$-DP algorithms, our lower bound holds naturally also for $(\eps, \delta)$-DP algorithms where $\delta$ is sufficiently small (i.e. $\delta = O(1/\eps)$). Thus, our results also show that allowing approximate-DP does not help for the DP anonymized histogram problem.

\definecolor{ashgrey}{rgb}{0.7, 0.75, 0.71}

\begin{table*}[t]
    \centering
\begin{tabular}{|c|c|c|c|c|c|}
    \cline{3-6}
      \multicolumn{2}{c|}{} & \multicolumn{2}{c|}{Previous Work} & \multicolumn{2}{c|}{This Note}\\
    \cline{3-6}
          \multicolumn{1}{c}{} &
          \multicolumn{1}{c|}{Bounds} & $\ell_1$ Error & Ref. & $\ell_1$ Error & Ref.\\ 
         \hline
         $\eps$-DP with $\eps \geq 1$ & Upper & $O(\sqrt{n} / e^{c \eps})$ for some $c > 0$ & \cite{Suresh19} & $O(\sqrt{n} / e^{\eps})$ & Theorem~\ref{thm:alg} \\ 
         (Low-Privacy) & Lower & $\Omega(\sqrt{n} / e^{\eps})$ & \cite{Suresh19} & - & -\\
      \hline
      $\eps$-DP with $\eps < 1$ & Upper & $O(\sqrt{n \log(1/\eps) / \eps})$ & \cite{Suresh19} & - & - \\
      (High-Privacy) & Lower & $\Omega(\sqrt{n / \eps})$ &  \cite{AldaS18} & $\Omega(\sqrt{n \log(1/\eps) / \eps})$ & Theorem~\ref{thm:lb-main} \\
    \hline
    \end{tabular}
     \caption{Summary of our results and previous results. Our lower bound also holds for $(\eps ,\delta)$-DP for any sufficiently small $\delta$ (depending only on $\eps$). The lower bounds in the high-privacy regime of our setting requires $\eps \geq \Omega(\log n/n)$ whereas that of~\cite{AldaS18} requires $\eps \geq \Omega(1/n)$; these are necessary because the error of $O(n)$ is trivial to achieve (by outputting the all-zero vector).}
    \label{table:summary}
  \end{table*}

\subsection{Proof Overview}
\label{sec:overview}

For the rest of this note, we will use the following equivalent definition of the problem in terms of privatizing integer partitions. An \emph{integer partition} of size $n$ is an infinite dimensional vector $\bp = (p_1, p_2, \dots)$ where $p_1 \geq p_2 \geq \cdots$ are non-negative integers and $\sum_{i \in \N} p_i = n$. Two integer partitions $\bp, \bp'$ are considered neighbors iff $\|\bp - \bp'\|_1 \leq 1$. The problem of DP integer partition of size at most $n$ is as follows: given an integer partition $\bp$ of size at most $n$, output a DP estimate $\bhp$ that minimizes the expected $\ell_1$-error, i.e. $\E[\|\bhp - \bp\|_1]$.

It is simple to see that the DP integer partition problem is equivalent to that of DP anonymized histogram, because an anonymized histogram can be viewed as an integer partition and vice versa\footnote{Given a size-$n$ integer partition $\bp = (p_i)_{i \in \N}$, we can create an input dataset where $p_i$ users have their inputs equal to $i$.} while preserving the neighboring relationship. Note that this equivalence is known and is also stated in~\cite{BlockiDB16,AldaS18}.

To describe the algorithm, we will also need an alternative representation of integer partitions: for a given integer partition $\bp$, its cumulative prevalence at $r \in \N$ is defined as $\phi_{\geq r}(\bp) := |\{i \in \N \mid p_i \geq r\}|$. For brevity, we also write $\bphi_\geq(\bp) := (\phi_{\geq 1}(\bp), \phi_{\geq 2}(\bp), \cdots)$. A crucial observation (shown in~\cite{Suresh19}) is that the $\ell_1$-distance in $\bp$ is the same as that of $\bphi_{\geq}$:

\begin{observation}[\cite{Suresh19}] \label{obs:l1-cum}
For any integer partitions $\bp, \bp'$, $\|\bp - \bp'\|_1 = \|\bphi_{\geq}(\bp) - \bphi_{\geq}(\bp')\|_1$.
\end{observation}

Below we give high-level overviews of our proofs. We will sometimes be informal for simplicity of presentation; all proofs will be formalized in later sections.

\paragraph{Algorithm for the Low-Privacy Regime.}
Suresh’s algorithm~\cite{Suresh19} for the low-privacy regime comes from the following observations. First, if we were told that \emph{all} of $p_i$’s are either at least $\sqrt{n}$ or equal to 0, then we can immediately arrive at an $\eps$-DP algorithm with $O(\sqrt{n}/e^\eps)$: by adding Geometric (aka Discrete Laplace) noise\footnote{See \Cref{sec:geom} for a more detail description of the distribution and its privacy guarantees.} to each of $p_1,\cdots , p_{\sqrt{n}}$. On the other hand, if
we were in the other extreme case where \emph{all} of $p_i$’s are less than $\sqrt{n}$, then we can also achieve a similar error guarantee: by adding Geometric noise to $\phi_{\geq 1}(\bp),  \cdots, \phi_{\geq \sqrt{n}}(\bp)$. In other words, each of the two extreme cases can be easily handled by either adding noise to just the counts themselves or to the cumulative prevalences. Of course, this does not constitute an algorithm for the general case yet. To handle the general inputs,~\cite{Suresh19} ``splits'' the partition into two parts based on the \emph{counts}. More precisely, a threshold $\tau \approx \sqrt{n}$ is selected and the integer partition is split into two parts: based on whether $p_i \geq \tau$. The ``high-count'' part is then handled by adding noise to the counts whereas the ``low-count'' part is handled by adding noise to the cumulative prevalences. At the end, the final output is then computed as the ``union'' of the two. Although this high-level idea is simple, it is \emph{not} differential private yet due to the fact that splitting based on count can result in a large sensitivity in each of the two parts: if $\bp$ has $p_i = \tau - 1$ but its neighbor $\bp'$ has $p'_i = \tau$, then the ``high-count'' part of $\bp$ and the ``high-count'' part of $\bp'$ differ by $\tau$ (in the $\ell_1$ distance) and this would have required a much larger amount of noise added to preserve privacy. To overcome this problem,~\cite{Suresh19} also makes the split itself differentially private by adding noise to the prevalences at $\tau$ and $\tau + 1$. This considerably complicates the algorithm--which now has to also add ``dummy counts'' at $\tau, \tau + 1$ to avoid the noisy prevalences becoming negative (these dummies need to be removed later)--and its privacy analysis--which requires novel tools like a ``dataset dependent composition theorem'' and several case analyses.

Our simple, and in hindsight obvious, observation is that this complication can be avoided entirely by splitting based on the \emph{rank}. That is, we simply let the first ``high-count'' part be $(p_1, \dots, p_{\sqrt{n}})$ and the second ``low-count'' part be $(p_{\sqrt{n} + 1},p_{\sqrt{n} + 2}, \dots)$. We can then proceed in a similar manner as~\cite{Suresh19}: add noise to the counts in the ``high-count'' part and to the cumulative prevalences in the ``low-count'' part, and then combine them to get the final answer. Here we do not encounter any issue with splitting at all and our analysis uses only elementary tools in differential privacy. Our algorithm is presented and analyzed in full detail in \Cref{sec:alg}.

We end by remarking that our ``split based on rank'' observation also simplifies the high-privacy algorithm of~\cite{Suresh19}. However, this does not change the asymptotic errors since a constant factor in $\eps$ only affects the error by a constant factor. Moreover, we stress that the splitting step is not the most complicated step in the high-privacy algorithm of~\cite{Suresh19}, which we will discuss more next.

\paragraph{Lower Bound for the High-Privacy Regime.}
To understand our lower bound, it is important to have a high-level understanding of the high-privacy algorithm of~\cite{Suresh19}. If one were to use the splitting and adding noise directly in this regime, the error incurred will be $O(\sqrt{n}/\eps)$ because the noise required to add to each count/cumulative prevalence is $O(1/\eps)$ in the case of $\eps < 1$.

To reduce the dependency of $\eps$ from $1/\eps$ to $\sqrt{\log(1/\eps)/\eps}$, Suresh~\cite{Suresh19} changes the second ``low-count'' part to consider $\phi_{\geq v}(\bp)$ for only $v \in V$ where $V \subseteq [\sqrt{n}/\eps]$ is such that its consecutive
elements have multiplicative gaps of roughly $1 + \sqrt{\frac{\log(1/\eps)}{\eps n}}$. The crucial points here are that (i) $(\phi_{\geq v}(\bp))_{v \in V}$ is enough to estimate the integer partition to within an error of $O\left(\sqrt{\frac{n \log(1/\eps)}{\eps}}\right)$ and (ii) the cardinality of $V$ is $O\left(\sqrt{\eps n \log(1/\eps)}\right)$, a multiplicative factor of $O(\sqrt{\eps \log(1/\eps)})$ smaller compared to $\sqrt{n}$. Intuitively speaking, (ii) also means that one can add less noise to them (compared to e.g. adding noise to all cumulative prevalences below $\sqrt{n}$). Formally exploiting (ii) however is somewhat complicated and is not needed for our lower bound; thus we will not go into detail here.

Now that we have a high level intuition of the algorithm, we briefly note that the lower bounds in~\cite{AldaS18,Suresh19} proceed by “encoding” each boolean vector $\bu \in \{0,1\}^m$ to an integer partition with a property that the distance of any two vectors are preserved (after appropriate scaling); this suffices to prove a lower bound due to a known strong lower bound for privatizing boolean vectors (see \Cref{lem:vec-dp-lb}). Our proof proceeds similarly but use a new encoding, guided by the above algorithm. Roughly speaking, we associate $[m]$ with the set $V$ and then start by encoding the all-zero vector to the elements of $V$ where $v \in V$ repeated a number of times proportional to $1/v$. If the associated coordinate in $u$ is one, we increase these values $v$’s to the next value in $V$. Otherwise, we keep them the same. This concludes our high-level overview of the proof; details are formalized in \Cref{sec:lb}.

We end this section by noting that our lower bound also gives a result on the size of packing of integer partitions under $\ell_1$ distance--arguably not surprising given the connection between packing and lower bounds in DP~\cite{HardtT10}--which nonetheless might be of independent interest. (See \Cref{app:packing}.)

\section{Notations and Preliminaries}

We use $\log$ to denote logarithm base 2 and, for any positive integer $m$, let $[m] := \{1, \cdots, m\}$ and $[m + 1:] := \N \setminus [m] = \{m + 1, \dots\}$. We use boldface letters for variables that represent vectors, and standard letters for those representing scalars. For a vector $\bv \in \R^{\cJ}$ and $\cJ' \subseteq \cJ$, we write $\bv_{\cJ'}$ to denote the restriction of $\bv$ on the coordinates in $\cJ'$, i.e. $(v_{j'})_{j' \in \cJ'}$. We often have to deal with vectors with trailing zeros; in this case, it is more convenient to write down just the non-zero part. To facilitate this, we define the following notion of $\ell_1$-distance for vectors of different dimensions: for $\bv \in \R^m, \bv' \in \R^{m'}$, let $\|\bv - \bv'\|_1 := \sum_{i \in [\max\{m,m'\}]} |z_i|$ where
\begin{align*}
z_i :=
\begin{cases}
v_i - v'_i &\text{ if } i \leq m, m', \\
v_i &\text{ if } m' < i \leq m, \\
-v'_i &\text{ if } m < i \leq m'.
\end{cases}
\end{align*}

\subsection{Integer Partition}

For every non-negative integer $n$, we use $\cI_n$ to denote the set of all size-$n$ integer partition $\bp$, i.e. those such that $\sum_i p_i = n$. We also use $\cI_{\leq n}$ to denote $\bigcup_{n' \leq n} \cI_{n'}$ and $\cI$ to denote $\bigcup_{n \in \N} \cI_n$. We say that two integer partitions $\bp$ and $\bp'$ are neighbors iff $\|\bp - \bp'\|_1 = 1$. For any $n \in \N$, the (expected $\ell_1$-)error of an algorithm $\cA$ for integer partition of size at most $n$ is defined\footnote{Note that in the main body of this note we assume that the algorithm $\cA$ knows that the input $\bp$ belongs to $\cI_{\leq n}$; please refer to \Cref{sec:open} and \Cref{app:unknown-n} for more discussion on this.} to be $\max_{\bp \in \cI_{\leq n}} \E[\|\cA(\bp) - \bp\|_1]$.

We define the \emph{union} of two partitions $\bp^1, \bp^2$, denoted by $\bp^1 \cup \bp^2$, in a natural manner i.e. $(\bp^1 \cup \bp^2)_i$ is defined as the $i$-th highest count in the multiset $\{p^1_i \mid i \in \N\} \cup \{p^2_i \mid i \in \N\}$. Note that this implies $\bphi_{\geq}(\bp^1 \cup \bp^2) = \bphi_{\geq}(\bp^1) + \bphi_{\geq}(\bp^2)$. This, together with \Cref{obs:l1-cum}, immediately yields:
\begin{observation} \label{obs:union-lower-error}
Let $\bp^1, \bp^2, \bq^1, \bq^2 \in \cI$. Then, $\|(\bp^1 \cup \bp^2) - (\bq^1 \cup \bq^2)\|_1 \leq \|\bp^1 - \bq^1\|_1 + \|\bp^2 - \bq^2\|_1$.
\end{observation}

\subsection{Differential Privacy}

Unless stated otherwise, the input to our algorithm will be integer-valued vectors and two datasets are consider neighbors iff their $\ell_1$-distance is at most one. We now recall a few additional tools from differential privacy.

\subsubsection{Group Privacy}

We say that two datasets $\cD, \cD'$ are $k$-neighboring if their exists a sequence $\cD_0 = \cD, \cD_1, \dots, \cD_{k - 1}, \cD_k = \cD'$ such that $\cD_{i - 1}, \cD_i$ are neighboring. Note that under the $\ell_1$ distance, two datasets $\bu, \bu'$ are $k$-neighboring iff $\|\bu - \bu'\|_1 \leq k$. We will use the following so-called \emph{group differential privacy} (see e.g. \cite[Fact 2.3]{SteinkeU16}) in our lower bound proof.

\begin{fact}[Group Privacy] \label{fact:group-dp}
Let $\cA$ be any $(\eps, \delta)$-DP algorithm. Then, for any $k$-neighboring database $\cD, \cD'$ and every subset $S$ of the output, we have $\Pr[\cA(\cD) \in S] \leq e^{k\eps} \cdot \Pr[\cA(\cD') \in S] + \frac{e^{k\eps} - 1}{e^{\eps} - 1} \cdot \delta$.
\end{fact}

\subsubsection{Geometric Mechanism}
\label{sec:geom}

We will employ the Geometric mechanism~\cite{GhoshRS12} in our algorithm. Recall that the (two-sided) Geometric distribution with parameter $\alpha < 1$, denoted by $\Geo(\alpha)$, has probability mass $\alpha^{|i|} \cdot \frac{1-\alpha}{1+\alpha}$ at each integer $i \in \Z$. Let $f$ be any function whose output is an integer-valued vector. The Geometric
mechanism (with parameter $\alpha$), simply adds a noise drawn from $\Geo(\alpha)$ to each coordinate of $f(\cD)$ independently. It is known that if the $\ell_1$-sensitivity of $f$, defined as $\max_{\text{neighboring } \cD,\cD'} \|f(\cD) - f(\cD')\|_1$, is at most one, then the Geometric mechanism with parameter $\alpha = e^{-\eps}$ is $\eps$-DP. Note that the $\Geo(e^{-\eps})$ distribution has expected absolute value of $O(1/\eps)$ for $0 < \eps \leq 1$ and $O(e^{-\eps})$ for $\eps > 1$.

\subsubsection{Lower Bound for Privatizing Boolean Vectors}

Our final tool is a lower bound on privatizing a boolean vector, which is stated formally below. (Recall that two vectors are neighbors iff their $\ell_1$-distance is at most one.) We note that a similar bound was proved already in~\cite{Suresh19} but only for $\delta = 0$; here we extend to $\delta > 0$ in a straightforward manner; for completeness, we provide its proof in Appendix A.

\begin{lemma} \label{lem:vec-dp-lb}
Let $\eps, \delta > 0$ and $m \in \N$. Then, for any $(\eps, \delta)$-DP algorithm $\cM: \{0, 1\}^m \to \{0, 1\}^m$, we have $\E_{\bz \sim \{0, 1\}^m}[\|\cM(\bz) - \bz\|_1] \geq e^{-\eps}m\cdot 0.5(1 - \delta)$.
\end{lemma}

\section{Algorithm in the Low-Privacy Regime}
\label{sec:alg}

Our algorithm is presented in Algorithm~\ref{alg:dph}. It follows the outline discussed in \Cref{sec:overview} except that we require a post-processing step since adding noise may make the counts/cumulative prevalences not corresponding to valid integer partitions. The post-processing simply selects the valid integer partitions that are closest to the noised counts/cumulative prevalences (Lines~\ref{l:pp1} and~\ref{l:pp2}). We note that this can be done efficiently using isotonic regression (see e.g.~\cite{Suresh19} for more detail).

\begin{algorithm}
\caption{DPAnonHist$_{\eps, n}(\bp)$}
\label{alg:dph}
\begin{algorithmic}[1]
\STATE \textbf{Input: } integer partition $\bp \in \cI_{\leq n}$, \textbf{Parameter: } $\eps \geq 1, n \in \N$
\STATE $m \leftarrow \lceil \sqrt{n} \rceil$
\STATE $\bp^{high,noised} = (p_1, \dots, p_m) + \Geo(e^{-\eps})^{\otimes m}$ \hfill\COMMENT{Add noise to high-count part.}
\STATE $\bphi_{\geq}^{low,noised} = \bphi_{\geq}(p_{m + 1}, \dots)_{[m]} + \Geo(e^{-\eps})^{\otimes m}$ \hfill\COMMENT{Add noise to low-count part.}
\STATE $\bp^{high,priv} \leftarrow \argmin_{\bq \in \cI_{\leq n}} \|\bp^{high,noised} - \bq\|_1$ \hfill\COMMENT{Post-process high-count part.} \label{l:pp1}
\STATE $\bp^{low,priv} \leftarrow \argmin_{\bq \in \cI_{\leq n}} \|\bphi_{\geq}^{low,noised} - \bphi_{\geq}(\bq)\|_1$ \hfill\COMMENT{Post-process low-count part.} \label{l:pp2}
\RETURN $\bp^{high,priv} \cup \bp^{low,priv}$.
\end{algorithmic}
\end{algorithm}

We now prove its privacy and utility guarantees.

\begin{theorem} \label{thm:alg}
For any $\eps \geq 1$ and $n \in \N$, $\text{DPAnonHist}_{\eps, n}$ is $\eps$-DP and its error is at most $O(\sqrt{n} / e^\eps)$ for integer partitions of size at most $n$.
\end{theorem}

\begin{proof}
For convenience, we define $\bp^{high} = (p_1, \dots, p_m), \bp^{low} = (p_{m + 1}, \dots)$ and $\bphi^{low}_{\geq} = \bphi_{\geq}(\bp^{low})_{[m]}$.

\paragraph{Privacy Analysis.} Notice that the final output is a post-processing of $(\bp^{high,noised}, \bphi_{\geq}^{low,noised})$. Hence, it suffices to show that $(\bp^{high,noised}, \bphi_{\geq}^{low,noised})$ satisfies $\eps$-DP. Observe also that $(\bp^{high,noised}, \bphi_{\geq}^{low,noised})$ is simply a result of the geometric mechanism $\Geo(e^{-\eps})$ applied on $(\bp^{high}, \bphi^{low}_{\geq})$; hence, in light of~\Cref{sec:geom}, it in turn suffices to show that the sensitivity of $\bp \mapsto (\bp^{high}, \bphi^{low}_{\geq})$ is at most one. To show this, let us consider neighboring datasets $\bp, \bp' \in \cI_{\leq n}$, i.e. where $\|\bp - \bp'\|_1 \leq 1$. We have
\begin{align*}
\|(\bp^{high}, \bphi^{low}_{\geq}) - (\bp'^{high}, \bphi'^{low}_{\geq})\|_1 
&= \|\bp_{[m]} - \bp'_{[m]}\|_1 + \|\bphi_{\geq}(\bp_{[m+1:]})_{[m]} - \bphi_{\geq}(\bp'_{[m+1:]})_{[m]}\|_1 \\
&\leq \|\bp_{[m]} - \bp'_{[m]}\|_1 + \|\bphi_{\geq}(\bp_{[m+1:]}) - \bphi_{\geq}(\bp'_{[m+1:]})\|_1 \\
(\text{\Cref{obs:l1-cum}}) &= \|\bp_{[m]} - \bp'_{[m]}\|_1 + \|\bp_{[m+1:]} - \bp'_{[m+1:]}\|_1 = \|\bp - \bp'\|_1 \leq 1,
\end{align*}
which concludes our proof that the algorithm is $\eps$-DP.

\paragraph{Utility Analysis.}
Observe that, since $m \geq \sqrt{n}$, we have $p_{m+1} \leq \frac{n}{m + 1} \leq m$ which implies that $\phi_{\geq r}(\bp^{low})$ for all $r > m$. From this, we have
\begin{align} \label{eq:truncation-free}
\|\bphi_{\geq}^{low} - \bphi_{\geq}(\bp^{low})\|_1 = 0.
\end{align}
Thus, we can derive
\begin{align*}
\|(\bp^{high,priv} \cup \bp^{low,priv}) - \bp\|_1 
&= \|(\bp^{high,priv} \cup \bp^{low,priv}) - (\bp^{high} \cup \bp^{low})\|_1 \\
(\text{From \Cref{obs:union-lower-error}}) &\leq \|\bp^{high,priv} - \bp^{high}\|_1 + \|\bp^{low,priv} - \bp^{low}\|_1 \\
(\text{From \Cref{obs:l1-cum}}) &= \|\bp^{high,priv} - \bp^{high}\|_1 + \|\bphi_{\geq}(\bp^{low,priv}) - \bphi_{\geq}(\bp^{low})\|_1 \\
(\text{From \Cref{eq:truncation-free}}) &= \|\bp^{high,priv} - \bp^{high}\|_1 + \|\bphi_{\geq}(\bp^{low,priv}) - \bphi_{\geq}^{low}\|_1 \\
(\text{Choices of } \bp^{low,priv}, \bp^{high,priv}) &\leq 2\|\bp^{high,noised} - \bp^{high}\|_1 + 2\|\bphi_{\geq}^{low,noised} - \bphi_{\geq}^{low}\|_1 \\
&\leq O(m/e^\eps) = O(\sqrt{n} / e^\eps),
\end{align*}
where the last inequality follows from the expected absolute value of $\Geo(e^{-\eps})$ (see~\Cref{sec:geom}).
\end{proof}

\section{Lower Bound in the High-Privacy Regime}
\label{sec:lb}

We will next prove the lower bound in the high-privacy regime, stated below. Here we only focus on sufficiently small $\eps < 10^{-6}$ to simplify our calculations but we stress that otherwise the lower bound of $\Omega(\sqrt{n}/e^\eps)$ (from \Cref{lem:lb-low-priv}) is already asymptotically tight. We also remark that the assumption $\eps \geq \Omega(\log n/n)$ is required; at $\eps = \Theta(\log n/n)$, our lower bound of $\Omega\left(\sqrt{\frac{n}{\eps} \log\left(\frac{1}{\eps}\right)}\right) = \Omega(n)$ is already the largest one can hope for and thus it cannot increase further even when $\eps$ gets smaller.

\begin{theorem} \label{thm:lb-main}
For any $n \in \N$ and $\frac{\log n}{n} < \eps < 10^{-6}, 0 < \delta < 0.1\eps$, there is no $(\eps, \delta)$-DP algorithm whose error is at most $o\left(\sqrt{\frac{n}{\eps} \log\left(\frac{1}{\eps}\right)}\right)$ for integer partitions of size at most $n$.
\end{theorem}

As stated earlier, our proof relies on an encoding of boolean vectors into integer partitions where the $\ell_1$ distance of a pair of vectors is preserved after encoding up to a constant factor, after an appropriate scaling. The main properties of the encoding are given below.

\begin{lemma} \label{lem:mapping}
Let $n, \Delta \in \N$ be such that $n \geq \Delta > 10\sqrt{n}$. There exists $m \geq \frac{1}{8} \cdot \frac{n}{\Delta}\log\left(\frac{\Delta}{\sqrt{n}}\right)$ and a mapping $\psi: \{0, 1\}^m \to \cI_{\leq n}$ such that $\frac{\|\psi(\by) - \psi(\bz)\|_1}{\|\by - \bz\|_1} \in [\frac{\Delta}{8m}, \frac{\Delta}{2m}]$ for all distinct $\by, \bz \in \{0, 1\}^m$.
\end{lemma}

Before we prove \Cref{lem:mapping}, let us show how we can simply use it to derive Theorem~\ref{thm:lb-main}.

\begin{proof}[Proof of Theorem~\ref{thm:lb-main}]
Suppose for the sake of contradiction that there is an $(\eps,\delta)$-DP algorithm $\cA$ whose error
is at most $0.005\Delta$ for integer partitions of size at most $n$ where $\Delta = \left\lceil \frac{1}{10} \sqrt{\frac{n}{\eps} \log\left(\frac{1}{\eps}\right)} \right\rceil$. It is not hard to verify (using the assumption on $\eps$) that $\Delta$ satisfies the precondition in \Cref{lem:mapping}; let $m, \psi$ be as in the lemma statement.

We define an algorithm $\cM$ for privatizing a vector $\bz \in \{0, 1\}^m$ as follows:
\begin{itemize}
\item Run $\cA$ on $\psi(\bz)$ to get $\bp \in \cI_{\leq n}$.
\item Output $\bz^* \in \{0, 1\}^m$ that minimizes $\|\psi(\bz^*) - \bp\|_1$.
\end{itemize}
We start by analyzing $\cM$'s utility. For $\bz \in \{0, 1\}^m$, we have
\begin{align*}
&\E[\|\bz^* - \bz\|_1] \\
(\text{From property of } \psi) &\leq \frac{8m}{\Delta} \cdot \E[\|\psi(\bz^*) - \psi(\bz)\|_1] \\
(\text{From the choice of } \bz^*) &\leq \frac{16m}{\Delta} \cdot \E[\|\cA(\psi(\bz)) - \psi(\bz)\|_1] \\
(\text{From the accuracy of } \cA) &\leq 0.08m.
\end{align*}
Furthermore, for any neighboring $\bz, \bz' \in \{0, 1\}^m$, the property of $\psi$ also implies
\begin{align*}
\|\psi(\bz) - \psi(\bz')\|_1 \leq \frac{\Delta}{2m} \|\bz - \bz'\|_1 \leq \frac{\Delta}{2m} \leq \frac{\Delta}{\frac{1}{4} \cdot \frac{n}{\Delta}\log\left(\frac{\Delta}{\sqrt{n}}\right)} = \frac{4\Delta^2}{n \log(\Delta / \sqrt{n})} \leq \frac{1}{\eps},
\end{align*}
where the last inequality follows from our choice of $\Delta$. Thus, from Fact~\ref{fact:group-dp}, we can conclude that $\cM$ is $(1, \delta')$-DP with $\delta' = \frac{e-1}{e^{\eps}-1} \cdot \delta \leq \frac{1.8}{\eps} \cdot \delta \leq 0.5$, where the latter follows from the assumption on $\delta$. The privacy and utility guarantees of $\cM$ then contradict with \Cref{lem:vec-dp-lb}.
\end{proof}

We now move on to prove \Cref{lem:mapping}. As outlined in \Cref{sec:overview}, our encoding assigns each coordinate with a representative count (called $p^\ell_r$ in the proof below) which are repeated so that the total count assigned to each coordinate is roughly similar. If this coordinate of the vector has value zero, then these counts remain the same. Otherwise, if this coordinate has value one, then we increase these counts so that their values are the same as the next count. The choices of the representative counts are inspired by that of the high-privacy algorithm of~\cite{Suresh19}. Specifically, we start with a count of roughly $\sqrt{n}$ and increase it by a multiplicative factor of roughly $1 + \frac{\Delta}{n}$ each time. Note that for simplicity of calculations we actually divide the representative counts into blocks and make the additive gap between two consecutive ones in each block equal.

To facilitate the proof, we also use the following notations: for a vector $\bu = (u_1, \dots, u_d) \in \R^d$ and $t \in \N$, let $\bu_{\times t} \in \R^{dt}$ denote the vector where each coordinate is repeated $t$ times in consecutive positions; more precisely, the $j$-th coordinate of $\bu_{\times t}$ is equal to $u_{1 + \lfloor(j-1)/d\rfloor}$. Furthermore, we use $\bu \circ \bu'$ to denote the concatenation of vectors $\bu$ and $\bu'$.

\begin{proof}[Proof of \Cref{lem:mapping}]
Let $L = \lfloor 0.5 \log \left(\Delta / \sqrt{n}\right) \rfloor, R = \lfloor n / \Delta \rfloor$ and $m = RL$; clearly, we have $m \geq \frac{1}{8} \cdot \frac{n}{\Delta}\log\left(\frac{\Delta}{\sqrt{n}}\right)$ as claimed. For every $\ell \in [L]$, let $s^\ell = \lfloor \frac{\Delta}{L 2^\ell} \rfloor, d^\ell = \lfloor s^\ell / R \rfloor$ and $p^\ell_r = s^\ell + (R - r) \cdot d^\ell$ for all $r \in \{0, 1, \dots, R\}$. For every vector $\bv \in \{0, 1\}^R$, we define the transformation $\psi^\ell: \{0, 1\}^R \to \N^R$ by
\begin{align*}
\psi^\ell(\bv)_r =
\begin{cases}
p^\ell_r &\text{ if } \bv_r = 0, \\
p^\ell_{r - 1} &\text{ if } \bv_r = 1.
\end{cases}
\end{align*}

Finally, we define $\psi: \{0, 1\}^{RL} \to \N^{(2^L - 1)R}$ by
\begin{align*}
\psi(\bz^1 \circ \bz^2 \circ \cdots \circ \cdots \bz^L) := \psi^1(\bz^1)_{\times 2^{0}} \circ \psi^2(\bz^2)_{\times 2^1} \cdots \circ \psi^L(\bz^L)_{\times 2^{L-1}}.
\end{align*}
where $\bz^\ell := \bz_{\{R(\ell - 1) + 1, \dots R\ell\}}$.

We claim that, for any $\bz \in \{0, 1\}^{RL}$, $\psi(\bz)$ belongs to $\cI_{\leq n}$. To see this, first observe by definition that $\psi(\bz)$ is non-decreasing. Furthermore, its total value is
\begin{align*}
\|\psi(\bz)\|_1
&= \sum_{\ell \in [L]} 2^{\ell - 1} \cdot \|\psi^\ell(\bz^\ell)\|_1 \\
(\text{From definition of } \psi^\ell) &\leq \sum_{\ell \in [L]} 2^{\ell - 1} \left(\sum_{r \in [R]} p^\ell_{r - 1}\right) \\
(\text{From definition of } p^\ell_r) &= \sum_{\ell \in [L]} 2^{\ell - 1} \left(\sum_{r \in [R]} s^\ell + (R  + 1 - r) \cdot d^\ell\right) \\
&\leq \sum_{\ell \in [L]} 2^{\ell - 1} \cdot R \cdot (s^\ell + R \cdot d^\ell) \\
(\text{From our choice of } d^\ell) &\leq \sum_{\ell \in [L]} 2^\ell \cdot R \cdot s^\ell \\
(\text{From our choice of } s^\ell) &\leq \sum_{\ell \in [L]} R\Delta / L \\
&= R \Delta \\
(\text{From our choice of } R) &\leq n.
\end{align*}

Next, note that for all $\ell \in [L]$, we have $\frac{\Delta}{L 2^\ell} \geq \frac{\Delta}{2^{2L}} \geq \sqrt{n} \geq 2n/\Delta$. This means that
\begin{align} \label{eq:rounding-err}
\frac{\Delta}{2L2^\ell} \leq s^\ell \leq \frac{\Delta}{L2^\ell} & &\text{ and } & & \frac{\Delta}{4 L 2^\ell R} \leq d^{\ell} \leq \frac{\Delta}{L 2^\ell R}.
\end{align}
Now, consider any $\by, \bz \in \{0, 1\}^{RL}$. We have
\begin{align*}
&\|\psi(\by) - \psi(\bz)\|_1 \\
(\text{From definition of } \psi) &= \sum_{\ell \in [L]} 2^{\ell - 1} \cdot \|\psi^\ell(\by^\ell) - \psi^\ell(\bz^\ell)\|_1 \\
(\text{From definition of } \psi^\ell) &= \sum_{\ell \in [L]} 2^{\ell - 1} \sum_{r \in [R]} d^\ell \cdot \ind[(\by^\ell)_r \ne (\bz^\ell)_r] \\
&= \sum_{\ell \in [L]} 2^{\ell - 1} d^\ell \sum_{r \in [R]} \ind[(\by^\ell)_r \ne (\bz^\ell)_r].
\end{align*}
From~\eqref{eq:rounding-err}, we have $2^{\ell - 1} d^\ell \in \left[\frac{\Delta}{8LR}, \frac{\Delta}{2LR}\right]$. This, together with the above inequality, implies that
\begin{align*}
\frac{\Delta}{8RL} \|\by - \bz\|_1 \leq \|\psi(\by) - \psi(\bz)\|_1 \leq \frac{\Delta}{2RL} \|\by - \bz\|_1,
\end{align*}
which concludes our proof.
\end{proof}

\section{Discussion and Open Questions}
\label{sec:open}

In this work, we give new upper and lower bounds for the DP anonymized histogram problem that asymptotically close the remaining gaps in literature. Nevertheless, there are several subtle variations to the problems that leave us with a few open questions.
\begin{itemize}
\item \textbf{Knowledge of the Algorithm on $n$}. In the main body of this note, we assume that the algorithm knows that the input $\bp$ belongs to $\cI_{\leq n}$ i.e. that it knows an upper bound $n$ of $\sum_{i} p_i$. This only makes the lower bound stronger, but it makes the algorithm (\Cref{alg:dph}) weaker. Fortunately, we show in Appendix B that the algorithm can be easily extended to the case where no upper bound on $\sum_{i} p_i$ is given, at the cost of an extra $\exp(-\Omega(n))$ term in the error, which is dominated by $O(\sqrt{n}/e^\eps)$ as long as $n \geq \Theta(\eps)$. It remains interesting whether the former term can be avoided in the case $n = o(\eps)$.
\item \textbf{Uniform vs Non-Uniform Lower Bound in terms of $n$.} Our lower bounds are proved against algorithms whose error guarantee must hold for all $\bp \in \cI_{\leq n}$. An interesting question is whether the lower bound remains true if we only require error guarantees for $\bp \in \cI_{n}$ (and the algorithm can perform arbitrarily bad for $\bp \in \cI_{\leq n - 1}$). For this version, it is not hard to show a similar lower bound but with a blow-up of $\eps$ by a factor of two (using a similar technique as in~\cite{Suresh19}, i.e. by having pairs of elements for each coordinate of the boolean vector, and increasing one element and decreasing another equally in the encoding). This means that our asymptotically tight lower bound in the high-privacy regime (Theorem~\ref{thm:lb-main}) remains the same. However, the lower bound for the low-privacy regime becomes $O(\sqrt{n}/e^{2\eps})$. Indeed, this latter bound was the one originally proved in~\cite{Suresh19} and we give a slightly modified proof that yields the lower bound of $\Omega(\sqrt{n} / e^\eps)$ but only against the stronger utility guarantee (for all $\bp \in \cI_{\leq n}$) in \Cref{app:lb-uniform}. It remains open whether the lower bound $O(\sqrt{n}/e^{\eps})$ holds if the error bound is only required to hold for all $\bp \in \cI_{n}$.
\end{itemize}

Since our work and all previous works focus on the \emph{central model} of differential privacy, it might be tempting to consider the more challenging \emph{local model}~\cite{KasiviswanathanLNRS11}. Unfortunately, the anonymized histogram problem generalizes the so-called \emph{Count Distinct} problem--where we simply would like to estimate the number of non-empty buckets--for which a strong lower bound of $\Omega(n)$ is known in the local model for any $\eps = O(1)$~\cite{ChenGKM21}. Nonetheless, the problem may still be interesting for other models of DP, such as Pan-Privacy~\cite{DworkNPRY10} or the Shuffle model~\cite{BittauEMMRLRKTS17,ErlingssonFMRTT19,CheuSUZZ19}. In both models, protocols with non-trivial error bounds of $O_{\eps}(\sqrt{n})$ are known for the Count Distinct problem~\cite{DworkNPRY10,MirMNW11,BalcerCJM21}. Can we also achieve a non-trivial error of $o(n)$ for the anonymized histogram problem when $\eps = O(1)$?

\bibliography{ref}
\bibliographystyle{alpha}

\appendix

\section{Proof of \Cref{lem:vec-dp-lb}}

Here we prove \Cref{lem:vec-dp-lb}; we stress that this prove is essentially a straightforward extension of the proof from \cite{Suresh19} for the case $\delta = 0$.

\begin{proof}[Proof of \Cref{lem:vec-dp-lb}]
For every $i \in [m]$, we use $\bz_{\neg i}$ to denote $(z_1, \dots, z_{i-1}, 1 - z_{i}, z_{i+1}, \dots, z_m)$. We have
\begin{align*}
&\E_{\bz \sim \{0, 1\}^m}[\|\cM(\bz) - \bz\|_1] \\
&= \sum_{i \in [m]} \E_{\bz \sim \{0, 1\}^m}[|\cM(\bz)_i - z_i|] \\
&= \sum_{i \in [m]} \Pr_{\bz \sim \{0, 1\}^m}[\cM(\bz)_i \ne z_i] \\
&= \sum_{i \in [m]} \frac{1}{2^m} \cdot \sum_{\bz \in \{0, 1\}^m} \Pr[\cM(\bz)_i \ne z_i] \\
&= \sum_{i \in [m]} \frac{1}{2^{m + 1}} \cdot \sum_{\bz \in \{0, 1\}^m} \left(\Pr[\cM(\bz)_i \ne z_i] + \Pr[\cM(\bz_{\neg i})_i = z_i]\right) \\
(\text{From } (\eps, \delta)\text{-DP guarantee of } \cM) &\geq \sum_{i \in [m]} \frac{1}{2^{m + 1}} \cdot \sum_{\bz \in \{0, 1\}^m} \left(\Pr[\cM(\bz)_i \ne z_i] + e^{-\eps}\left(\Pr[\cM(\bz)_i = z_i] - \delta\right)\right) \\
&\geq \sum_{i \in [m]} \frac{1}{2^{m + 1}} \cdot \sum_{\bz \in \{0, 1\}^m} e^{-\eps}\left(\Pr[\cM(\bz)_i \ne z_i] + \Pr[\cM(\bz)_i = z_i] - \delta\right) \\
&= \sum_{i \in [m]} \frac{1}{2^{m + 1}} \cdot \sum_{\bz \in \{0, 1\}^m} e^{-\eps}\left(1 - \delta\right) \\
&= e^{-\eps} m \cdot 0.5(1 - \delta). 
\end{align*}
\end{proof}

\section{Algorithm for the Case of Unknown $n$}
\label{app:unknown-n}

In this section, we extend our algorithm (\Cref{alg:dph}) to the case where we do not know that $\bp \in \cI_{\leq n}$. We will only consider $\eps \geq 2$ here since otherwise we may run e.g. the 1-DP algorithm from~\cite{Suresh19}. For this case of $\eps \geq 2$, we spend a fixed $\eps$ budget of 1 to estimate the $n = \sum_i p_i$. We then run \Cref{alg:dph} with a slightly inflated value $n'$ of estimated $n$. The point here is that if $n' \geq n$ we may use the utility guarantee of \Cref{alg:dph} to argue about the expected error. Otherwise, if $n' < n$, we may only bound the error trivially as $O(n)$. Nonetheless, the latter case happens only rarely, i.e. with probability $\exp(-\Omega(n))$, which allows us to derive the final error bound. The full algorithm and its analysis are presented below.

\begin{algorithm}
\caption{DPAnonHist$_{\eps}(\bp)$}
\begin{algorithmic}[1]
\STATE \textbf{Input: } integer partition $\bp \in \cI$, \textbf{Parameter: } $\eps \geq 2$
\STATE $n \leftarrow \sum_{i=1}^{\infty} p_i$.
\STATE $\hn \leftarrow n + \Geo(1/e)$.
\STATE $n' \leftarrow 2 \max\{1, \hn\}$.
\STATE $\bp^* \leftarrow$ DPAnonHist$_{\eps - 1, n'}(\bp)$
\RETURN $\bp' \leftarrow \argmin_{\bq \in \cI_{\leq n'}} \|\bq - \bp^*\|_1$.
\end{algorithmic}
\end{algorithm}

\begin{lemma} \label{lem:alg-no-n}
For any $\eps \geq 2$, DPAnonHist$_\eps$ is $\eps$-DP and its expected $\ell_1$-error on any input $\bp \in \cI$ is at most $$O(\sqrt{n}/e^\eps) + \exp(-\Omega(n)).$$
\end{lemma}

We remark that the first term in the error guarantee above dominates as long as $n \geq \Omega(\eps)$.

\begin{proof}[Proof of \Cref{lem:alg-no-n}]
To see its privacy guarantee, observe that $\hn$ is a result of applying 1-DP Geometric mechanism. Thus, by basic composition of DP and $(\eps-1)$-DP guarantee\footnote{Note that the DP guarantee of DPAnonHist$_{\eps,n}$ holds even when the input $\bp$ does not belong to $\cI_{\leq n}$.} of DPAnonHist$_{\eps-1,n'}$, we can conclude that DPAnonHist$_{\eps}$ is $\eps$-DP as desired.

Next, we prove its utility guarantee. To do this, consider two cases:
\begin{itemize}
\item $n' \geq n$. In this case, we have
\begin{align*}
\E[\|\bp' - \bp\|_1] \leq \E[2\|\bp^* - \bp\|_1] \leq O(\sqrt{n'}/e^\eps),
\end{align*}
where the last inequality follows from the utility guarantee of DPAnonHist$_{\eps-1,n'}$.
\item $n' < n$. In this case, due to the choice of $\bp'$, the error is at most $O(n)$.
\end{itemize}
We can then combine the two cases as follows:
\begin{align*}
\E[\|\bp' - \bp\|_1]
&= \sum_{m \in \N} \Pr[n' = m] \cdot \E[\|\bp' - \bp\|_1 \mid n' = m] \\
&\leq \Pr[n' < n] \cdot O(n) + \sum_{m \geq n} \Pr[n' = m] \cdot O(\sqrt{m} / e^\eps) \\
&\leq \Pr[n' < n] \cdot O(n) + \sum_{m \geq n} \Pr[n' = m] \cdot O((\sqrt{2n} + \sqrt{\max(0, m - 2n)}) / e^\eps) \\
&= O(\sqrt{n}/e^\eps) + \Pr[n' < n] \cdot O(n) + \sum_{m \geq 2n} \Pr[n' = m] \cdot O(\sqrt{m - 2n} / e^\eps) \\
&= O(\sqrt{n}/e^\eps) + \Pr[\Geo(1/e) < -n/2] \cdot O(n) + \sum_{m \geq 2n} \Pr[\Geo(1/e) = m/2 - n] \cdot O(\sqrt{m - 2n} / e^\eps) \\
&= O(\sqrt{n}/e^\eps) + \Pr[\Geo(1/e) < -n/2] \cdot O(n) + \sum_{i \in \N} \Pr[\Geo(1/e) = i] \cdot O(\sqrt{i} / e^\eps) \\
&\leq (\sqrt{n} / e^\eps) + \exp(-\Omega(n)) + O(1/e^\eps),
\end{align*}
where, in the last inequality, we use the tail bound of $\Geo(1/e)$ to bound $\Pr[\Geo(1/e) < -n/2]$.
\end{proof}

\section{Lower Bound in the Low-Privacy Regime}
\label{app:lb-uniform}

In this section, we prove a lower bound of $\Omega(\sqrt{n}/e^\eps)$ on the error of $(\eps, \delta)$-DP algorithm for integer partition of size at most $n$. As stated earlier, this extended the result of~\cite{Suresh19} for the case $\delta = 0$. Furthermore, it also quantitatively improved upon the lower bound of~\cite{Suresh19} which is only $\Omega(\sqrt{n}/e^{2\eps})$; this however is not due to a technical innovation but rather a slight difference in the lower bound (i.e. for $\bp \in \cI_n$ vs for $\bp \in \cI_{\leq n}$). We refer the readers to \Cref{sec:open} for more discussion.

\begin{lemma} \label{lem:lb-low-priv}
For any $\eps > 0$ and $\delta \in [0, 0.9)$, no $(\eps, \delta)$-DP algorithm for integer partition of size at most $n$ has error $o(\sqrt{n}/e^{\eps})$.
\end{lemma}

\begin{proof}
Let $m = \lfloor \sqrt{n} \rfloor$. Consider $\psi: \{0,1\}^m \rightarrow \Z^m_{\geq 0}$ defined by $\psi(\bu)_i =2(m-i) + u_i$.

It is simple to verify that $\psi(\bu) \in \cI_{\leq n}$ for all $\bu \in \{0, 1\}^m$ and that $\|\psi(\bu)-\psi(\bu')\|_1 = \|\bu - \bu'\|_1$. As a result, if we have an $(\eps, \delta)$-DP algorithm for $n$-partition has an expected $\ell_1$-error of $o(\sqrt{n}/e^{\eps})$, then it would give an $(\eps, \delta)$-DP algorithm for privatizing $\{0,1\}^m$ with $\ell_1$ error of $o(m/e^\eps)$, contradicting \Cref{lem:vec-dp-lb}.
\end{proof}

\section{Implication on Packing of Integer Partitions Under $\ell_1$ Distance}
\label{app:packing}

Finally, we state below the corollary of our lower bound proof technique (\Cref{sec:lb}) on the size of packings of integer partitions under the $\ell_1$-distance. We remark that the bound on the size of the packing $T$ in~\Cref{cor:packing} is tight up to a constant factor in the exponent; otherwise, using techniques from~\cite{HardtT10}, we would have arrived at a stronger lower bound for $\eps$-DP algorithm for integer partition which would contradict with the upper bound in~\cite{Suresh19}.

\begin{corollary} \label{cor:packing}
Let $n, \Delta \in \N$ be such that $n \geq \Delta > 10\sqrt{n}$. There exists integer partitions $\bp_1, \dots, \bp_T \in \cI_{\leq n}$ such that $\|\bp_i - \bp_j\|_1 \in [0.01\Delta, \Delta]$ for all $i \ne j$ and
\begin{align*}
T = \exp\left(\Omega\left(\frac{n}{\Delta}\log\left(\frac{\Delta}{\sqrt{n}}\right)\right)\right).
\end{align*}
\end{corollary}

To prove~\Cref{cor:packing}, we will use the following simple fact which is equivalent to the existence of error correcting code with constant rate and 0.1 (relative) distance.

\begin{lemma} \label{lem:ecc}
For any $m \in \N$, there exist $\bu_1, \dots, \bu_T \in \{0, 1\}^m$ where $T = \exp(\Omega(m))$ such that $\|\bu_i - \bu_j\|_1 \geq 0.1m$ for all $i \ne j$.
\end{lemma}

\begin{proof}[Proof of \Cref{cor:packing}]
Let $m, \psi$ be as in \Cref{lem:mapping}. Now, let $\bu_1, \dots, \bu_T \in \{0, 1\}^m$ be the vectors as guaranteed by \Cref{lem:ecc} where $T = \exp(\Omega(m)) = \exp\left(\Omega\left(\frac{n}{\Delta}\log\left(\frac{\Delta}{\sqrt{n}}\right)\right)\right)$. Let $\bp_i := \psi(\bu_i)$ for $i \in [T]$.
From how we select $\bu_1, \dots, \bu_T$ and the property of $\psi$, we can conclude that $0.01\Delta \leq \|\bp_i - \bp_j\|_1 \leq \Delta$ as desired.
\end{proof}

\end{document}